%% file: Main.tex
\documentclass[11pt,letterpaper]{article}
\usepackage[letterpaper, margin=1in]{geometry}
\usepackage{amssymb}
\usepackage{amsmath}
\usepackage{wasysym}
\usepackage{graphicx}
\usepackage[usenames,dvipsnames,svgnames,table]{xcolor}
\usepackage{bbm}
\usepackage[compact]{titlesec}

 \newtheorem{theorem}{Theorem}[section]
\newtheorem{observation}[theorem]{Observation}

 \newtheorem{lemma}[theorem]{Lemma}

 \def\squarebox#1{\hbox to #1{\hfill\vbox to #1{\vfill}}}
 
 \newcommand{\qed}{\hspace*{\fill}
 	\vbox{\hrule\hbox{\vrule\squarebox{.667em}\vrule}\hrule}\smallskip}
 \newenvironment{proof}{\begin{trivlist}
 		\item[\hspace{\labelsep}{\bf\noindent Proof: }]
 	}{\qed\end{trivlist}}
 
 \newenvironment{proof-synopsis}{\begin{trivlist}
 		\item[\hspace{\labelsep}{\bf\noindent Proof (Synopsis): }]
 	}{\qed\end{trivlist}}

\renewcommand{\phi}{\varphi}

\newcommand{\set}[1]{\left\{#1\right\}}
\newcommand{\NN}{{\mathbbm{N}}}
\newcommand{\ZZ}{{\mathbbm{Z}}}
\newcommand{\RR}{{\mathbbm{R}}}
\newcommand{\CC}{{\mathbbm{C}}}
\newcommand{\QQ}{{\mathbbm{Q}}}
\renewcommand{\AA}{{\mathbbm{A}}}

\newcommand{\B}{{\cal B}}
\newcommand{\A}{{\cal A}}

\newcommand{\co}[1]{\overline{#1}}

\newcommand{\re}{{\rm Re}}
\newcommand{\im}{{\rm Im}}

\newcommand{\norm}[1]{{\left\lVert#1\right\rVert}}
\newcommand{\heig}[1]{{H(#1)}}

\newcommand{\relBow}{\mathrel{\bowtie}}

\newcommand{\PSPACE}{{\bf PSPACE}}
\newcommand{\NPh}{{\bf NP-Hard }}
\newcommand{\posslp}{{\bf PosSLP}}
\newcommand{\NPposslp}{$\text{{\bf NP}}^\text{{\bf PosSLP}}$}
\newcommand{\NPtoRP}{$\text{{\bf NP}}^\text{{\bf RP}} $ }

\newcommand{\vect}[2]{\begin{pmatrix}
		#1 \\ #2
	\end{pmatrix}}

\title{The Polytope-Collision Problem}
\author{Shaull Almagor\\Department of Computer Science\\Oxford University, UK\\shaull.almagor@cs.ox.ac.uk \and Jo\"el Ouaknine\\Department of Computer Science\\Oxford University, UK\\joel@cs.ox.ac.uk\and James Worrell\\Department of Computer Science\\Oxford University, UK\\jbw@cs.ox.ac.uk}
\date{}
\begin{document}
\maketitle
\begin{abstract}
	The {\em Orbit Problem} consists of determining, given a matrix $\A\in \RR^{d\times d}$ and vectors $x,y\in \RR^d$, whether there exists $n\in \NN$ such that $\A^n=y$. This problem was shown to be decidable in a seminal work of Kannan and Lipton in the 1980s. Subsequently, Kannan and Lipton noted that the Orbit Problem becomes considerably harder when the {\em target} $y$ is replaced with a subspace of $\RR^d$. Recently, it was shown that the problem is decidable for vector-space targets of dimension at most three, followed by another development showing that the problem is in $\PSPACE$ for polytope targets of dimension at most three.
	
	In this work, we take a dual look at the problem, and consider the case where the {\em initial} vector $x$ is replaced with a polytope $P_1$, and the target is a polytope $P_2$. Then, the question is whether there exists $n\in \NN$ such that $\A^nP_1\cap P_2\neq \emptyset$. We show that the problem can be decided in $\PSPACE$ for dimension at most three. As in previous works, decidability in the case of higher dimensions is left open, as the problem is known to be hard for long-standing number-theoretic open problems.
	
	Our proof begins by formulating the problem as the satisfiability of a parametrized family of sentences in the existential first-order theory of real-closed fields. Then, after removing quantifiers, we are left with instances of simultaneous positivity of sums of exponentials. Using techniques from transcendental number theory, and separation bounds on algebraic numbers, we are able to solve such instances in $\PSPACE$.
\end{abstract}
\pagebreak
\input{intro.tex}
\input{tools.tex}
\input{php2system.tex}

\input{solveSystem.tex}
\input{Conclusions.tex}

\appendix

\pagebreak
\input{appendix}

\small
\bibliographystyle{plain}
\bibliography{PolytoPoly3D}
\normalsize

\end{document}

%% file: intro.tex
\section{Introduction}
\label{sec: problem intro}
Given a linear transformation $\A$ over the vector space $\RR^d$, together with a starting point $x$, the {\em orbit} of $x$ under $\A$ is the infinite sequence $x,\A x,\A^2x,\ldots$. A natural decision problem in discrete linear dynamical systems is whether the orbit of $x$ ever hits a particular target set $V$ (assuming suitable, effective representations of $\A$, $x$, and $V$). An early instance of this problem was raised by
Harrison in 1969~\cite{harrison1969lectures} for the special case in which $V$ is simply a point in $\RR^d$. Decidability remained open for over ten years, and was finally settled in a seminal paper of Kannan and Lipton, who moreover gave a polynomial-time decision procedure~\cite{kannan1980orbit}. In subsequent
work~\cite{kannan1986polynomial}, Kannan and Lipton noted that the Orbit Problem becomes considerably harder when the target $V$ is replaced by a subspace of $\RR^d$: indeed, if $V$ has dimension $d-1$, the problem is equivalent to the {\em Skolem Problem}, known to be \NPh but whose decidability has remained open for over 80 years~\cite{tao2008structure}.
However, for low-dimensional target spaces, the Orbit Problem becomes more tractable. Indeed, it was recently shown in~\cite{chonev2013orbit} that the problem is decidable for vector-space targets of dimension at most three, with polynomial-time complexity for one-dimensional targets, and complexity in \NPtoRP for two- and three-dimensional targets. Another development followed in~\cite{chonev2015polyhedron}, where the authors consider more intricate target sets, namely polytopes. It is shown in~\cite{chonev2015polyhedron} that up to dimension three, the problem can be solved in $\PSPACE$. In addition, it is shown that for higher dimensions, the problem becomes hard with respect to long-standing number-theoretic open problems.

A key motivation for studying the Orbit Problem comes from program verification, particularly the problem of determining
whether a simple while loop with affine assignments and guards will terminate or not. 
Similar reachability questions were considered and left open by Lee and Yannakakis in~\cite{lee1992online} for what they termed “real affine transition systems”. Similarly, decidability for the case of a single-halfspace target was mentioned as an open problem by Braverman in~\cite{braverman2006termination}. 

An important aspect of termination problems for linear loops is the quantification of the initial point. Traditionally, the `Termination problem' in the program-verification literature (see, e.g.~\cite{ben2014ranking}) refers to termination of while loops for {\em all} possible initial starting points. In~\cite{ouaknine2015termination} the traditional Termination Problem is solved over the integers for while loops, assuming diagonalisability of the associated linear transformation. To our knowledge, very little else is known on the general problem of universally quantified inputs. In contrast, the works in~\cite{chonev2013orbit,chonev2015polyhedron} study the termination problem where the input is fixed (but the target space is complicated). This corresponds to verifying the termination of a concrete run of a linear loop. It should be noted that the techniques used for analyzing the latter differ significantly from the former.

In this work, we take a dual look at the problem, and study the case where the input is existentially quantified. Thus, we are given a set $P_1\subseteq \RR^d$, and a target set $P_2$, and the problem is to decide whether there exists $x\in P_1$ and $n\in \NN$ such that $\A^n x\in P_2$. In practice, this corresponds to deciding {\em safety} properties of linear loops: we think of $P_2$ as some error set, and the problem is to decide whether there exists an input that would cause the program to reach the error set. 

Specifically, the focus of this paper is the {\bf 3D Polytope-Collision Problem} ({\bf 3DPCP}, for short): Given two polytopes $P_1$ and $P_2$ in $\RR^3$ (represented as an intersection of halfspaces) and a matrix with real-algebraic entries\footnote{We denote by $\AA$ the set of algebraic numbers.} $\A\in (\AA\cap \RR)^{3\times 3}$, determine whether there exists a point $x\in P_1$ and a natural number $n$ such that $\A^n x\in P_2$. 

We present the following effectiveness result on the 3D Polytope-Collision Problem.
\begin{theorem}
	\label{thm:main PSPACE}
	 3DPCP is decidable in $\PSPACE$.
\end{theorem}
Note that as proved in~\cite{chonev2015polyhedron}, when the dimension is at least four, the polytope-collision problem becomes hard with respect to number-theoretic open problems. 

Before describing our approach, we explain why this result is somewhat surprising. Consider a simplification of 3DPCP, where the initial polytope $P$ is a segment between points $x$ and $y$, and we wish to decide whether the orbit of $P$ under the matrix $\A$ collides with another polytope $R$. We can represent $P$ as the single point $(x, y)$ in $\RR^6$, and extend $\A$ to a matrix $\B\in \RR^{6\times 6}$ that has two copies of $\A$ on its diagonal. Then, the orbit of $P$ under $\A$ corresponds to the orbit of $(x,y)$ under $\B$.  However, the respective target space in $\RR^6$ becomes the set of all points $(u,v)$ such that the line between $u$ and $v$ in $\RR^3$ intersects $R$. While this is a semi-algebraic set, it is quite complicated, and recall that the polytope hitting problem is already hard in dimension four. Thus, this approach suggests that the problem may be as hard as the hitting problem in $\RR^6$. 

Our approach to proving Theorem~\ref{thm:main PSPACE} is as follows. Observe that 3DPCP can be formulated as the problem of deciding whether there exists $n\in \NN$ such that $\A^nP_1$ intersects $P_2$ (where $\A^nP_1=\set{\A^nx:x\in P_1}$). In Section~\ref{sec:3DPHP to system} we reduce this formulation of 3DPCP to the problem of solving a system of inequalities, as we now describe.

In Section~\ref{sec:polyhedra intersection} we identify two types of intersection of 3D polytopes, namely (1) where a vertex of one polytope lies in the other polytope, and (2) where an edge of one polytope intersects a face of the other polytope. We show that under a certain representation, an intersection of polytopes is always of one of these types.
Note that while each of these types seems symmetric with respect to the two polytopes, in our setting the polytopes have an inherent asymmetry, as $\A^nP_1$ is dependent on $n$ whereas $P_2$ is not. 

In order to overcome this asymmetry, in Section~\ref{sec:reduction to invertible} we reduce 3DPCP to the case where the matrix $\A$ is invertible. Then, considering $\A^n P_1$ and $P_2$ is symmetric to considering $P_1$ and $(\A^{-1})^n P_2$. 

Next, in Section~\ref{sec:from invertible to equations} we observe that intersections of Type (1) can be decided using the work in \cite{chonev2015polyhedron}, and we are left to address intersections of Type (2). We formulate this type of intersection as $\exists n\in\NN\ \Phi(\alpha^n,\co{\alpha}^n,\rho^n)$, where $\Phi$ is a sentence in the existential first-order theory of real-closed fields, and $\alpha,\co{\alpha}$, and $\rho$ are the eigenvalues of the matrix $\A$, with $\alpha\in \AA\setminus\RR$ and $\rho\in \AA\cap \RR$ (the case where $\A$ has only real eigenvalues is simpler, and  we handle it in Appendix~\ref{apx:real eigen}). Moreover, $\Phi$ contains only linear expressions (with respect to its variables, where $n$ is treated as a constant), and at most three real variables. We proceed by eliminating the quantifiers from $\Phi$. 
We use the fact that the expressions in $\Phi(n)$ are linear to apply the simple Fourier-Motzkin quantifier-elimination algorithm~\cite{fourier1826solution}. We note that while other quantifier-elimination algorithms (e.g., \cite{renegar1992computational}) offer better asymptotic complexity, since the number of variables in $\Phi$ is constant, Fourier-Motzkin elimination takes polynomial time. Moreover, its simplicity allows us to keep track of the expressions in the quantifier free equivalent of $\Phi(n)$. Specifically, we show that this output consists of a disjunction of systems, where each system is a conjunction of expressions of the form 
\begin{equation}
\label{eq:intro expression}
A\alpha^{2n}+\co{A}\co{\alpha}^{2n}+ B \alpha^n\rho^n + \co{B}\co{\alpha}^n\rho^n + C\rho^{2n} + D|\alpha|^{2n}+ E\alpha^n + \co{E}\co{\alpha}^n+ F\rho^n + G\relBow 0
\end{equation} where $\mbox{$\bowtie$}\in\set{>,=}$.

Finally, Section~\ref{sec:solving the system} is the heart of our technical contribution, in which we show how to solve such systems. Intuitively, we normalize Expression~(\ref{eq:intro expression}) such that the maximal modulus of its terms is $1$, thus obtaining an expression of the form $A\gamma^{2n}+\co{A}\co{\gamma}^{2n}+ B \gamma^n + \co{B}\co{\gamma}^n + C + r(n)\relBow 0$ with $|\gamma|=1$ and $r(n)$ tending exponentially fast to $0$. We then consider two cases, depending on whether $\gamma$ is a root of unity or not. If $\gamma$ is a root of unity, we show that it is enough to consider polynomially many expressions with only real elements, which can be handled using relatively standard techniques. If $\gamma$ is not a root of unity, things are more involved. Then, by utilizing consequences of the Baker-W\"ustholz theorem~\cite{baker1993logarithmic}, we are able to show that the expression $|A\gamma^{2n}+\co{A}\co{\gamma}^{2n}+ B \gamma^n + \co{B}\co{\gamma}^n + C|$ is bounded away from $0$ by an inverse polynomial in $n$. Then, using a separation bound due to Mignotte~\cite{mignotte1983some}, we show that $r(n)$ decays fast enough to obtain a bound $N\in \NN$ such that $r(n)$ does not affect the sign of $A\gamma^{2n}+\co{A}\co{\gamma}^{2n}+ B \gamma^n + \co{B}\co{\gamma}^n + C$ for all $n>\NN$. Finally, since $\gamma$ is not a root of unity, it is dense in the unit circle, and we can replace the analysis of the former expression by analysis of the simpler function $f(z)=Az^{2}+\co{A}\co{z}^{2}+ B z + \co{B}\co{z}+C$ on the unity circle, from which we obtain our main result.

%% file: tools.tex
\section{Mathematical Tools}
In this section we introduce the key technical tools used in this paper.
\subsection{Algebraic numbers}
\label{sec:algebraic numbers}
For $p\in \ZZ[x]$ a polynomial with integer coefficients we denote by $\norm{p}$ the bit length of its representation as a list of coefficients encoded in binary. 
Note that the {\em degree} of $p$, denoted $\deg(p)$ is at most $\norm{p}$, and the {\em height} of $p$ --- i.e., the maximum of the absolute values of its coefficients, denoted $\heig{p}$ --- is at most $2^\norm{p}$.

We begin by summarising some basic facts about
algebraic numbers (denoted $\AA$) and their (efficient) manipulation.
The main references include~\cite{basu2005algorithms,cohen2013course, renegar1992computational}.
A complex number $\alpha$ is {\em algebraic} if it is a root of a single-variable polynomial with integer coefficients. The
{\em defining polynomial} of $\alpha$, denoted $p_\alpha$, is the unique polynomial of least degree, and whose coefficients do not have common factors, which vanishes at $\alpha$. The {\em degree} and {\em height} of $\alpha$ are respectively those of $p$, and are denoted $\deg(\alpha)$ and $\heig{\alpha}$. A standard representation\footnote{Note that this representation is not unique.} for algebraic numbers is to encode $\alpha$ as a tuple comprising its defining polynomial
together with rational approximations of its real and imaginary parts of sufficient precision to distinguish $\alpha$ from the other roots of $p_\alpha$. More precisely, $\alpha$ can be represented by $(p_\alpha, a, b, r)\in \ZZ[x]\times \QQ^3$ provided that $\alpha$ is the unique root of $p_\alpha$ inside the circle in $\CC$ of radius $r$ centred at $a + bi$. A  separation bound due to Mignotte~\cite{mignotte1983some} asserts that for roots $\alpha\neq \beta$ of a polynomial $p\in \ZZ[x]$,  we have
\vspace*{-2pt}
\begin{equation}
\label{eq:Mignotte}
|\alpha-\beta|>\frac{\sqrt{6}}{d^{(d+1)/2}H^{d-1}}
\end{equation}
where $d=\deg(p)$ and $H=\heig{p}$. Thus if $r$ is required to be less than a quarter of the root-separation bound, the representation is well-defined and allows for equality checking.
Given a polynomial $p\in \ZZ[x]$, it is well-known how to compute standard representations of each of its roots in time  polynomial in $\norm{p}$~\cite{basu2005algorithms,cohen2013course,pan1996optimal}. Thus given an algebraic number $\alpha$ for which we have (or wish to compute) a
standard representation, we write $\norm{\alpha}$ to denote the bit length of this representation. From now on, when referring to computations on algebraic numbers, we always implicitly refer to their standard representations.

Note that Equation~\ref{eq:Mignotte} can be used more generally to separate arbitrary algebraic numbers: indeed, two algebraic numbers $\alpha$ and $\beta$ are always roots of the polynomial $p_\alpha p_\beta$ of degree at most $\deg(\alpha)+\deg(\beta)$, and of height at most $\heig{\alpha}\heig{\beta}$. 
Given algebraic numbers $\alpha$ and $\beta$, one can compute
$\alpha+\beta$, $\alpha\beta$, $1/\alpha$ (for $\alpha\neq 0$), $\co{\alpha}$, and $|\alpha|$, all of which are algebraic, in time polynomial in $\norm{\alpha}+\norm{\beta}$. Likewise, it is straightforward to check whether $\alpha=\beta$.
Moreover, if $\alpha\in\RR$, deciding whether $\alpha>0$ can be done in time polynomial in $\norm{\alpha}$. Efficient  algorithms for all these tasks can be found in~\cite{basu2005algorithms,cohen2013course}.

\subsection{First-order theory of the reals}
\label{sec:first order theory}
Let $\overrightarrow{x}=x_1,\ldots, x_m$ be a list of $m$ real-valued variables, and let $\sigma(\overrightarrow{x})$ be a Boolean combination of atomic predicates of the form $g(\overrightarrow{x}) \relBow 0$, where each $g(\overrightarrow{x})\in \ZZ[x]$ is a polynomial with integer coefficients over these variables, and $\relBow\in \set{>,=}$. A {\em sentence of the first-order theory of the reals} is of the form 
$Q_1x_1 Q_2x_2\cdots Q_mx_m \sigma(\overrightarrow{x})$,
where each $Q_i$ is one of the quantifiers $\exists$ or $\forall$. Let us denote the above formula by $\tau$ , and write $\norm{\tau}$ to denote the bit length of its syntactic representation.
Tarski famously showed that the first-order theory of the reals is decidable~\cite{tarski1951decision}. His procedure, however, has non-elementary complexity. Many substantial improvements followed over the years, starting with Collins’s technique of cylindrical algebraic decomposition~\cite{collins1975quantifier}, and culminating with the fine-grained analysis of Renegar~\cite{renegar1992computational}. In this paper, we focus exclusively on the situation in which the number of variables is uniformly bounded.
\begin{theorem}[Renegar]
	\label{thm:renegar}
	Let $M\in \NN$ be fixed, let $\tau$ be of the form $Q_1x_1 Q_2x_2\cdots Q_mx_m \sigma(\overrightarrow{x})$. Assume that the number of 	variables in $\tau$ is bounded by $M$ (i.e., $m \le M$). Then the truth value of $\tau$ can be determined in time polynomial in $\norm{\tau}$.
\end{theorem}
An important property of the first-order theory of the reals is that it admits {\em quantifier elimination}. That is, consider two lists of variables $\overrightarrow{x},\overrightarrow{y}$ and a sentence $Q_1x_1\cdots Q_mx_m  \sigma(\overrightarrow{x},\overrightarrow{y})$ with the variables of $\overrightarrow{y}$ being free, then there exists an (unquantified) sentence $\sigma'(\overrightarrow{y})$ such that for every assignment $\pi$ to the variables in $\overrightarrow{y}$ it holds that $\sigma'(\pi)$ is true iff $Q_1x_1\cdots Q_mx_m \sigma(\overrightarrow{x},\pi)$ is true. 

When the polynomials in $\sigma$ are all linear and the quantifiers are all existential, then quantifier elimination can be performed using the Fourier-Motzkin quantifier-elimination algorithm~\cite{fourier1826solution} (see Appendix~\ref{apx:quantifier elimination} for details). The benefit of this algorithm is its simplicity, which allows us to remove quantifiers symbolically.

We remark that algebraic constants can also be incorporated as coefficients in the first-order theory of the reals, as follows. Consider a polynomial $g(x_1,\ldots,x_m)$ with algebraic coefficients $c_1,\ldots, c_k$. We replace every $c_i$ with a new, existentially-quantified variable $y_i$, and add to the sentence the predicates $p_{c_i}(y_i)=0$ and $(y_i-(a+bi))^2< r^2$, where $(p_{c_i},a,b,r)$ is the representation of $c_i$. Then, in any evaluation of this formula to True, it must hold that $y_i$ is assigned value $c_i$.

\subsection{Polytopes and their representation}
\label{sec:tools polytopes}
A {\em polytope} $P$ in $\RR^3$ is an intersection of finitely many halfspaces in $\RR^3$: $P=\{x\in \RR^3: v_1^T x\ge c_1\wedge\ldots\wedge v_k^T x\ge c_k\}$ for vectors $v_1,\ldots,v_k\in \RR^3$ and numbers $c_1,\ldots,c_k\in \RR$. The {\em halfspace description} of $P$ is then $(v_1,c_1),\ldots,(v_k,c_k)$. When all entries are algebraic, we denote by $\norm{P}$ the description length.

The {\em dimension} of a polytope $P$, denoted $\dim(P)$, is the dimension of the subspace of $\RR^3$ spanned by $P$. The dimension of $P$ can be computed in time polynomial in $\norm{P}$ by solving polynomially many linear programs.
In $\RR^3$, the dimension of a polytope is in $\set{0,\ldots,3}$. A 2D boundry of a 3D polytope is a 2D polytope called a {\em face}. Similarly, the boundries of 2D polytopes (and in particular of faces) are called {\em edges}, and the boundries of edges are {\em vertices}. Every 3D polytope, except the trivial $\RR^3$ and $\emptyset$, has at least one face (but not necessarily edges or vertices). Since vertices and edges are crucial for our algorithms, we present the following lemma from~\cite{chonev2015polyhedron}
\begin{lemma}[\cite{chonev2015polyhedron} Lemma A.1]
	\label{lem:2d polytope decomposition}
	Suppose $P\subseteq \RR^3$ is a 2D polytope. Then $P=\bigcup_{i=1}^m A_i$, where $m$ is finite and each $A_i$ is of the form $A_i=\set{u_i+\alpha v_i+\beta w_i: T_i(\alpha,\beta)}$ where $u_i,v_i,w_i\in \RR^3$ and the predicates $T_i(\alpha,\beta)$ are from the following:
	\begin{itemize}
		\setlength\itemsep{-1pt}
		\item $T_i(\alpha,\beta)\equiv \alpha\ge 0 \wedge \beta\ge 0$ ($A_i$ is an infinite cone)
		\item $T_i(\alpha,\beta)\equiv \alpha \ge 0 \wedge \beta\ge 0 \wedge \alpha+\beta\le 1$ ($A_i$ is a triangle)
		\item $T_i(\alpha,\beta)\equiv \alpha\ge 0 \wedge \beta\ge 0 \wedge \beta \le 1$ ($A_i$ is an infinite strip)
	\end{itemize}	
	
	Furthermore, if we are given a halfspace description of $P$ with length $\norm{P}$, the size of the representation of each
	vector $u_i,v_i,w_i$ is at most $\norm{P}^{O(1)}$.
\end{lemma}
Note that since the representation of $u_i,v_i$, and $w_i$ is polynomial, it follows that $m$ is at most exponential in $\norm{P}$, and moreover, that iterating over the sets $A_i$ can be done in $\PSPACE$.

%% file: php2system.tex
\section{From 3DPCP to a System of Inequalities}
\label{sec:3DPHP to system}
In this section we reduce 3DPCP to the problem of solving a system of inequalities. More precisely, we show how to solve 3DPCP be solving an exponential number of systems of equalities and inequalities, and that iterating over these systems can be done in $\PSPACE$. In Section~\ref{sec:solving the system} we tackle the main technical challenge of solving each such system in $\PSPACE$, thus concluding the proof of Theorem~\ref{thm:main PSPACE}. 

As mentioned in Section~\ref{sec: problem intro}, we start by studying the intersection of polytopes.

\subsection{Intersection of polytopes}
\label{sec:polyhedra intersection}
Consider two intersecting polytopes $Q_1$ and $Q_2$ in $\RR^3$. In this section, we characterize the intersection of $Q_1$ and $Q_2$, which would later simplify the solution of 3DPCP. To illustrate the idea, assume that both $Q_1$ and $Q_2$ are bounded 3D polytopes. In this case, we can assume w.l.o.g.\ that $Q_1$ and $Q_2$ are both tetrahedra. Indeed, every bounded 3D polytope with $d$ vertices can be decomposed into a union of at most $\binom{d}{4}$ tetrahedra, and two such decompositions intersect iff two of the tetrahedra in the respective decompositions intersect. Under this assumption, there are two possible ``types'' of intersections: either $Q_1$ is contained in $Q_2$ (or vice-versa), or an edge of $Q_1$ intersects a face of $Q_2$ (or vice-versa). When the polytopes are bounded, we can relax the first requirement, and require instead that a vertex of $Q_1$ lies in $Q_2$ (or vice-versa).

In general, however, $Q_1$ or $Q_2$ may be unbounded. In this case we need to be slightly more careful. Indeed, as stated in Section~\ref{sec:tools polytopes}, unbounded polytopes might have no vertices or edges, but only faces (unless the polytope is $\RR^3$ or $\emptyset$, in which case the problem is trivial). For example, consider the case where $Q_1$ and $Q_2$ are infinite prisms. Then, it is possible that $Q_1\cap Q_2\neq \emptyset$ and neither are contained in each other, but no edge of $Q_1$ intersects a face of $Q_2$ (and vice-versa). 

Therefore, to get the above characterization for unbounded polytopes, we need to add ``fictive'' edges. Since we assume the input polytopes are non trivial, then each of them has at least one face, and recall that the faces of a 3D polytope are 2D polytopes. 
By employing Lemma~\ref{lem:2d polytope decomposition} on the faces of the polytopes, we get that each face of $Q_1$ and of $Q_2$ can be written as $\bigcup_{i=1}^mA_i$ as per Lemma~\ref{lem:2d polytope decomposition}. 
Observe that every set $A_i$ in the decomposition of Lemma~\ref{lem:2d polytope decomposition} has at least two edges and one vertex, and that a non-empty intersection $A_i\cap A'_j$ in such decompositions also intersects an edge of at least one of the two sets (the only involved case is the intersection of two infinite strips, where one should notice that the strips are only infinite to one side).

We conclude that the above characterization of the intersection of polytopes is correct also for unbounded ones. In the following, when we refer to a vertex/edge of an unbounded polytope, we mean the vertices and edges of the sets in the decomposition of Lemma~\ref{lem:2d polytope decomposition}. 

Thus, we have that $Q_1$ intersects $Q_2$ if at least one of the following holds:\\
1. There exists a vertex of $Q_1$ that is in $Q_2$. \hspace*{1cm}  3. An edge of $Q_1$ intersects a face of $Q_2$.\\
2. There exists a vertex of $Q_2$ that is in $Q_1$. \hspace*{1cm}  4. An edge of $Q_2$ intersects a face of $Q_1$.\\

\subsection{Reduction to the invertible case}
\label{sec:reduction to invertible}
In the notations of Section~\ref{sec:polyhedra intersection}, we wish to check the intersection of $Q_1=\A^n P_1$ and $Q_2=P_2$ for an existentially quantified $n\in \NN$. As mentioned in Section~\ref{sec: problem intro}, if $\A$ is invertible, then the problem is symmetric with respect to $Q_1$ and $Q_2$. Indeed, $\A^n P_1$ intersects $P_2$ iff $P_1$ intersects $(\A^{-1})^n P_2$. However, if $\A$ is not invertible, the problem is not clearly symmetric. 
In this section, we reduce 3DPCP to the case where $\A$ is an invertible matrix. 

Consider polytopes $P,R\subseteq \RR^3$, and let $\A\in (\AA\cap \RR)^{3\times 3}$ be a singular matrix, so $0$ is an eigenvalue of $\A$.
Consider first the case where the multiplicity of $0$ is $1$. 
Thus, we can write $\A=D^{-1} \begin{pmatrix}
0 & 0 \\
0 & B \\
\end{pmatrix} D$ where $D$ is an invertible matrix with real-algebraic entries, and $B\in (\AA\cap \RR)^{2\times 2}$. Indeed, if $\A$ has only real eigenvalues then this is achieved by converting $\A$ to Jordan form, and if $\A$ has complex eigenvalues $\alpha$ and $\co{\alpha}$, then this is achieved by setting $D=(v,u,w)$ where $v$ is an eigenvector corresponding to $0$, and $u+iw$ is an eigenvector corresponding to $\alpha$. In addition, $B$ is invertible, since its eigenvalues are the nonzero eigenvalues of $\A$.

In Appendix~\ref{apx:reduction to invertible}, we show that in this case, there exist polytopes $P',R'\subseteq \RR^2$ such that for every $n\ge 2$ the following holds: there exists $x\in P$ such that $\A^nx\in R$ iff there exists $x'\in P'$ such that $B^{n-1}x'\in R'$. Thus, it is enough to consider the polytopes $P',R'$ and the invertible matrix $B$. Moreover, we show that computing $P'$ and $R'$ can be done in polynomial time. We also show a similar approach can be taken when $0$ has multiplicity $2$ or $3$ (with the latter being trivial, since $\A$ is then nilpotent).


It should be noted that in the reduction above, even if the input had only rational entries, the output may still require a real-algebraic description. However, the degree and height of the algebraic numbers involved in the description of the output polytopes remain polynomial in the size of the input. 

Finally, we note that we can always {\em increase} the dimension of the problem while maintaining an invertible matrix. Indeed, Given a invertible matrix $B\in (\AA\cap \RR)^{2\times 2}$,
we can consider the invertible matrix $\begin{pmatrix}
1 & 0 \\
0 & B \\
\end{pmatrix}$,
and change $P,R\subseteq \RR^2$ to $\set{1}\times P,\set{1}\times R\subseteq \RR^3$ (and a similar approach when $B\in (\AA\cap \RR)^{1\times 1}$).
Thus, it is enough to solve the problem in the invertible case in dimension 3. 

\subsection{From the invertible case to an equation system}
\label{sec:from invertible to equations}
In this section we focus on solving 3DPCP in the invertible case. 

Let $P_1,P_2$ be the input polytopes (whose description may contain algebraic numbers, as per the reduction of Section~\ref{sec:reduction to invertible}), and let $\A\in (\AA\cap \RR)^{3\times 3}$ be an invertible matrix.
By Section~\ref{sec:polyhedra intersection}, and since $\A$ is invertible, it suffices to decide whether there exists a number $n\in \NN$ such that either there exists a vertex $x$ of $P_1$ with $\A^n x\in P_2$, or whether there exists an edge $e$ of $P_1$ such that $\A^n e$ intersects a face of $P_2$. Note that we may need to reverse the roles of $P_1$ and $P_2$, and use $\A^{-1}$ instead of $\A$. We remark that $\norm{\A^{-1}}$ is polynomial in $\norm{\A}$, and moreover --- since the eigenvalues of $\A^{-1}$ are inverses of those of $\A$ --- the description length of the eigenvalues of $\A^{-1}$ is equal to that of $\A$.

In \cite{chonev2015polyhedron}, the authors show that the problem of deciding, given a polyhedron $P$ in $\RR^3$, a vector $x\in\RR^3$, and a matrix $\A\in (\AA\cap \RR)^{3\times 3}$, whether there exists $n\in \NN$ such that $\A^n x\in P$ is solvable in \PSPACE. This solves the former case. It remains to solve the latter.

We thus assume that we are given as input a matrix $\A\in (\AA\cap \RR)^{3\times 3}$, an edge $E=\set{u+\lambda v: \lambda\in J}$ where $u,v\in \RR^3$ and $J$ is either $[0,1]$ or $[0,\infty)$, and a face $F=\set{s+\mu t+ \nu r: T(\mu,\nu)}$, where $s,t,r\in \RR^3$ and $T(\mu,\nu)$ is one of the following predicates (as per Lemma~\ref{lem:2d polytope decomposition}):
\begin{itemize}
\setlength\itemsep{-3pt}
\item $T(\mu,\nu)\equiv \mu\ge 0\wedge \nu\ge 0$
\item $T(\mu,\nu)\equiv \mu\ge 0\wedge \nu\ge 0\wedge \mu+\nu\le 1$
\item $T(\mu,\nu)\equiv \mu\ge 0\wedge \nu\ge 0\wedge \nu\le 1$
\end{itemize}

We wish to determine whether there exists a number $n$ and $x\in E$ such that $\A^n x\in F$. In the following, we will treat the case where $E=\set{u+\lambda v: \lambda\in [0,1]}$ and $F=\{s+\mu t+\nu r: \mu\ge 0\wedge \nu\ge 0\wedge \mu+\nu\le 1\}$. The other cases are slightly simpler, and can be solved {\em mutatis-mutandis}.

Consider the eigenvalues of $\A$. Since $\A$ is a $3\times 3$ invertible matrix, either all the eigenvalues are real, or there is one real eigenvalue $\rho$, and two complex, conjugate eigenvalues, $\alpha$ and $\co\alpha$. In the latter case, $\A$ is also diagonalizable. We consider here the latter case. In Appendix~\ref{apx:real eigen} we show how to handle the former case, which is easier.

Thus, let us assume that the eigenvalues of $\A$ are $\rho\in \AA\cap \RR$ and $\alpha,\co{\alpha}\in \AA$. We can compute an invertible matrix $B\in \AA^{3\times 3}$ such that $A=B^{-1}\begin{pmatrix}
\rho & 0 & 0\\
0 & \alpha & 0\\
0& 0 & \co{\alpha}\\
\end{pmatrix}B$, and the rows of $B$ are the respective eigenvectors. Note that if $w_\alpha$ is an eigenvector of $\alpha$, then $\co{w_\alpha}$ is eigenvector of $\co{\alpha}$, so we can write $B=\begin{pmatrix}
 w_\rho & w_{\alpha} & \co{w_\alpha}
\end{pmatrix}^T$. We now have that $\A^n= B^{-1}\begin{pmatrix}
\rho^n & 0 & 0\\
0 & \alpha^n & 0\\
0& 0 & \co{\alpha}^n\\
\end{pmatrix}B$ for every $n\in \NN$. By analyzing the structure of $B$ and $B^{-1}$, it is not hard to verify that every entry of $\A^n$ is a linear combination of $\alpha^n,\co{\alpha}^n$ and $\rho^n$ such that the coefficients of  $\alpha^n$ and $\co{\alpha}^n$ are conjugates, and the coefficient of $\rho^n$ is real. That is, 
for every $1\le i,j\le 3$ it holds that $(\A^n)_{i,j}=c_{i,j}\alpha^n+ \co{c_{i,j}}\,\co{\alpha}^n+d_{i,j}\rho^n$ for coefficients $c_{i,j}\in \AA$ and $d_{i,j}\in \AA\cap \RR$ (independent of $n$).

Consider a vector $x=u+\lambda v\in E$. We can write $\A^nx=\A^n u+\lambda \A^n v$, and observe that for $1\le i\le 3$ we have $(\A^nu)_i=(c_{i,1} u_1 +c_{i,2}u_2 +c_{i,3}u_3)\alpha^n+ \co{(c_{i,1} u_1 +c_{i,2}u_2 +c_{i,3}u_3)}\,\co{\alpha}^n+(d_{i,1}u_1+d_{i,2}u_2+d_{i,3}u_3)\rho^n$, and a similar structure holds for $\A^n v$. By renaming the coefficients, we can write
$(\A^n u+\lambda \A^n v)_i=f_i\alpha^n +\co{f_i}\, \co{\alpha}^n + g_i \rho^n + \lambda (h_i\alpha^n +\co{h_i}\, \co{\alpha}^n + k_i\rho^n)$ where $f_i,h_i\in \AA$ and $g_i,k_i\in \AA\cap\RR$ for $1\le i\le 3$.

We can now formulate the problem as follows: does there exists a number $n\in \NN$ such that the following first-order sentence is true: $\exists \lambda,\mu,\nu:0\le \lambda,\mu,\nu\le 1 \wedge \mu+\nu\le 1 \wedge$
\begin{equation}
\label{eq:quantified}
  \bigwedge_{i=1}^3 \left( f_i\alpha^n +\co{f_i}\, \co{\alpha}^n + g_i \rho^n + \lambda (h_i\alpha^n +\co{h_i}\, \co{\alpha}^n + k_i\rho^n)= s_i +\mu t_i +\nu r_i\right)
\end{equation}

As mentioned in Section~\ref{sec:first order theory}, we can convert (\ref{eq:quantified}) to an equivalent, quantifier-free sentence. Since our reasoning requires this equivalent sentence to have a special structure, we must explicitly remove the quantifiers. This is done in Appendix~\ref{apx:quantifier elimination} using Fourier-Motzkin quantifier elimination~\cite{fourier1826solution}, where we conclude the following.
\begin{theorem}
	\label{thm:systems}
	There exist constants $M,M'$ such that the sentence (\ref{eq:quantified}) is equivalent to a disjunction $\bigvee_{i=1}^M {\rm Sys}_i$ where for every $1\le i\le M$, $\rm Sys_i$ is a conjunction of at most $M'$ expressions of the form 
	\begin{equation}
	\label{eq: generic expression}
	A\alpha^{2n}+\co{A}\co{\alpha}^{2n}+ B \alpha^n\rho^n + \co{B}\co{\alpha}^n\rho^n + C\rho^{2n} + D|\alpha|^{2n}+ E\alpha^n + \co{E}\co{\alpha}^n+ F\rho^n + G\relBow 0
	\end{equation} where $\mbox{$\bowtie$}\in\set{>,=}$, $A,B,E\in \AA$, and $C,D,F,G\in \AA\cap\RR$.
	Moreover, the description of ${\rm Sys}$ is polynomial in $\norm{I}$.
\end{theorem}

%% file: solveSystem.tex
\section{Solving the System}
\label{sec:solving the system}
This section constitutes the main technical challenge of the paper, namely to decide whether there exists $n\in \NN$ such that the disjunction presented in Theorem~\ref{thm:systems} is true. We refer to such an $n$ as a {\em solution} for the disjunction.

We first note that it is enough to consider each system in the disjunction separately. Indeed, since 
the number of systems is bounded, independent of the input,
we can try to solve each one separately. Our goal is then to decide, given a system {\rm Sys} of expressions as per Theorem~\ref{thm:systems}, whether there exists a solution $n\in \NN$ that satisfies all the expressions simultaneously. 

We divide our analysis to two cases. First we handle the (straightforward) case where $\frac{\alpha}{|\alpha|}$ is a root of unity. We then proceed to consider the more involved case, where $\frac{\alpha}{|\alpha|}$ is not a root of unity.

\subsection{The case where $\frac{\alpha}{|\alpha|}$ is a root of unity}
\label{sec:root of unity}
Suppose that $\frac{\alpha}{|\alpha|}$, denoted $\gamma$, is a root of unity. We can now treat $(\ref{eq: generic expression})$ as
$$|\alpha|^{2n}A\gamma^{2n}+|\alpha|^{2n}\co{A}\co{\gamma}^{2n}+ |\alpha|^{n}B\gamma^n\rho^n + |\alpha|^{n}\co{B}\co{\gamma}^n\rho^n + C\rho^{2n} +  D|\alpha|^{2n}+ |\alpha|^n E\gamma^n + |\alpha|^n\co{E}\co{\gamma}^n+ F\rho^n + G\relBow  0$$
Let $d$ be the order of $\gamma$, then $\gamma^2$ is also a root of unity of order at most $d$. Thus, there are at most $d^2$ possible values for $(\gamma^n,\gamma^{2n})$, determined by the pair $(n\bmod d,2n \bmod d)$. 
We can now treat the expression as $d^2$ expressions of real-algebraic sums of exponentials. We show that $d\le \deg(\gamma)^2$, so these can be solved in $\PSPACE$ using standard techniques of asymptotic analysis, by considering the coefficients and the moduli of $\alpha$ and $\rho$ (see Appendix~\ref{apx:root of unity} for details).

\subsection{The case where $\frac{\alpha}{|\alpha|}$ is not a root of unity}
\label{sec: not root of unity}
When $\gamma=\frac{\alpha}{|\alpha|}$ is not a root of unity, things are more involved. Nonetheless, we prove the following theorem.
\begin{theorem}
	\label{thm:solve system}
	The problem of deciding whether a system {\rm Sys} of expressions of the form $(\ref{eq: generic expression})$ has a solution, is in \PSPACE.
\end{theorem}

Before proving the theorem, we need some definitions. 
In the following, we assume w.l.o.g.\ that $\rho>0$. Indeed, if $\rho<0$ then we can divide into two cases according to the parity of $n$, and solve each separately (note that $\rho\neq 0$ since the matrix $\A$ is invertible).

For an expression of the form $(\ref{eq: generic expression})$, we obtain its {\em normalized expression} by dividing it by $(\max\{|\alpha|^2,|\alpha|\rho,\rho^2,|\alpha|,|\rho|\})^n$ (and such that the coefficient of the element we divide by is nonzero). Thus, the normalized expression is of the form
\begin{equation}
\label{eq:normalized}
A\gamma^{2n}+\co{A}\co{\gamma}^{2n}+ B \gamma^n + \co{B}\co{\gamma}^n + C + r(n)\relBow 0,
\end{equation}
with $\gamma\in \AA$ such that $|\gamma|=1$ and $\gamma$ is not a root of unity, $A,B\in \AA$ and $C\in \AA\cap \RR$ are not all $0$, and
$r(n)=\sum_{l=1}^m D_l \beta^n_l+ \co{D_l}\co \beta^n_l$, where $|\beta_l|<1$ for every $1\le l\le m$, and $0\le m\le 4$ (note that for uniformity we treat real numbers in $r(n)$ as a sum of complex conjugates). 
For every $1\le l\le m$, $\beta_l$ is a quotient of two elements from the set $\set{\alpha,\alpha^2,\rho,\rho^2,\alpha\rho}$. Since $\alpha$ and $\rho$ are eigenvalues of $\A$, $\deg(\alpha),\deg(\rho)$ are $\norm{\A}^{O(1)}$. Thus, by Section~\ref{sec:algebraic numbers}, $\deg(\beta_l)=\norm{\A}^{O(1)}$, and $\heig{\beta_l}=2^{\norm{\A}^{O(1)}}$.

Since $\gamma$ is not a root of unity, then $\set{\gamma^n:n\in \NN}$ is dense in the unit circle. With this motivation in mind, we define, for a normalized expression, its {\em dominant function} $f:\CC\to\RR$ as $f(z)=Az^2+\co{A}\co{z}^{2}+ B z + \co{B}\co{z} + C$. Observe that $(\ref{eq:normalized})$ is now equivalent to $f(\gamma^n)+r(n)\relBow 0$.

The following lemma is our main technical tool in proving Theorem~\ref{thm:solve system}.

	\begin{lemma}
		\label{lem:main lemma}

		Consider a normalized expression as in $(\ref{eq:normalized})$. Let $\norm{I}$ be its encoding length, and let $f$ be its dominant function.
		Then there exists $N\in \NN$ computable in polynomial time in $\norm{I}$ with $N=2^{\norm{I}^{O(1)}}$
		such that for every $n>N$ it holds that 
		\begin{enumerate}
			\setlength\itemsep{-3pt}
			\item $f(\gamma^n)\neq 0$,
			\item $f(\gamma^n)>0$ iff $f(\gamma^n)+r(n)>0$,
			\item $f(\gamma^n)<0$ iff $f(\gamma^n)+r(n)<0$.
		\end{enumerate}
	\end{lemma}
	In particular, the lemma implies that if $f(n)+r(n)=0$, then $n\le N$. 
	The proof of Lemma~\ref{lem:main lemma} relies on the following lemma from~\cite{ouaknine2014ultimate}, which is itself a consequence of the Baker-W\"ustholz Theorem~\cite{baker1993logarithmic}.
	\begin{lemma}[\cite{ouaknine2014ultimate}]
		\label{lem:Baker on unit circle}
		There exists $D\in \NN$ such that for all algebraic numbers $\zeta,\xi$ of modulus 1, and for every $n\ge 2$, if $\zeta^n\neq \xi$, then $|\zeta^n-\xi|>\frac{1}{n^{(\norm{\zeta}+\norm{\xi})^D}}$.
	\end{lemma}
	We now turn to prove Lemma~\ref{lem:main lemma}. The following synopsis contains the main ideas. The full proof can be found in Appendix~\ref{apx:proof main lemma}.
\begin{proof-synopsis}
	Since $\set{\gamma^n:n\in \NN}$ is dense on the unit circle, we consider $f(z)$ for $z$ in the unit circle. In the full proof, we show that $\set{z:f(z)=0\wedge |z|=1}$ contains at most four points $\set{z_1,\ldots,z_4}$, whose coordinates are algebraic. 	Since $\gamma$ is not a root of unity, it holds that $\gamma^{n_1}\neq \gamma^{n_2}$ for every $n_1\neq n_2\in \NN$ . 
	Thus, there exists $N_1\in\NN$ such that $\gamma^n\notin\set{z_1,\ldots,z_4}$ for every $n>N_1$. Moreover, by Lemma D.1 in~\cite{chonev2013complexity}, we have that $N_1=k^{O(1)}$, where  $k= \norm{\gamma}+\sum_{j=1}^4 \norm{z_j}$, and $N_1$ can be computed in polynomial time in $k$.
	Then, by Lemma~\ref{lem:Baker on unit circle}, there exists a constant $D\in \NN$ such that for every $n\ge N_1$ and $1\le j\le 4$ we have that $|\gamma^n-z_j|>\frac{1}{n^{(k^D)}}$. Intuitively, for $n>N_1$ we have that $\gamma^n$ does not get close to any $z_i$ ``too quickly'' as a function of $n$. In particular, for $n>N_1$ we have $f(\gamma^n)\neq 0$. It thus remains to show that for large enough $n$, $r(n)$ does not affect the sign of $f(\gamma^n)+r(n)$. Intuitively, this is the case because $r(n)$ decreases exponentially, while $|f(\gamma^n)|$ is bounded from below by an inverse polynomial. While proving that this holds in general is not very difficult, note that we also need the bound on $N$ in the statement of the Lemma to be effectively computable and to be $2^{\norm{I}^{O(1)}}$, which complicates things significantly.
		
	We consider the function $g:(-\pi,\pi]\to \RR$ defined by $g(x)=f(e^{ix})$.  Explicitly, we have $g(x)=2|A|\cos(2x+\theta_A)+2|B|\cos(x+\theta_B)+C$ where $\theta_A=\arg(A)$ and $\theta_B=\arg(B)$. By the above, $g$ has at most four roots, denoted $\phi_1,\ldots,\phi_4$. We now show that there exist $N_2\in \NN$ and a non-negative polynomial $p(n)$ such that $f(\gamma^n)=g(\arg (\gamma^n))>\frac{1}{p(n)}$ for every $n>N_2$. For every $1\le j\le 4$ consider the first non-zero Taylor polynomial $T_j$ of $g$ around $\phi_j$. In Lemma~\ref{lem:degree of taylor} we show that the degree of such approximations is at most $3$. We show that there exists $\epsilon_1>0$ such that for every  $x\in(\phi_j-\epsilon_1,\phi_j+\epsilon)$ it holds that (1) $|g(x)-T_j(x)|\le \frac12 |T_j(x)|$, (2) $g$ is monotone on either side of $\phi_j$, and (3) $T$ is monotone with the same tendency of $g$  (see Figure~\ref{fig:functions} in Appendix~\ref{apx:proof main lemma} for an illustration).
	In Lemma~\ref{lem:epsilon computable} we also show that crucially, we can require $\epsilon_1$ to be efficiently computable and $\frac1\epsilon=2^{n^{O(1)}}$.
	
    Consider $n\in \NN$ such that $\gamma^n\in \bigcup_{j=1}^4 (\phi_j-\epsilon_1,\phi_j+\epsilon_1)$ and such that $n>N_1$, then as we have seen above, $\frac{1}{n^{(k^D)}}<|\gamma^n-z_j|$. But $|\gamma^n-z_j|<|\arg(\gamma^n)-\phi_j|$ (since the euclidean distance is smaller than the arc length), so $|\arg(\gamma^n)-\phi_j|>\frac{1}{n^{(k^D)}}$. From requirements (1) and (2) of $\epsilon_1$, we get that $|g(\arg(\gamma^n))|\ge \frac12|T_j(\gamma^n)|$ and from the monotonicity of $T_j$ in the neighbourhood of $\phi_j$ (requirement (3)), we have that $\frac12|T_j(\gamma^n)|>\frac12 \min\set{{|T_j(\phi_j+\frac{1}{n^{(k^D)}})|,|T_j(\phi_j-\frac{1}{n^{(k^D)}})|}}$, from which we conclude that $|g(\arg(\gamma^n))|>\frac{1}{p(n)}$ for some non-negative polynomial $p$. Moreover, we can compute the representation of $p$ in polynomial time.
    
    Finally, for $x\notin \bigcup_{j=1}^4 (\phi_j-\epsilon_1,\phi_j+\epsilon_1)$, we have that $|g(x)|$ is bounded from below by a constant. Our careful accounting of $\norm{\epsilon_1}$ in Lemma~\ref{lem:epsilon computable} allows us to compute this bound, and show that it is not too small.
    
    The last step in the proof is to show that $r(n)$ decreases fast enough such that $r(n)<\frac{1}{p(n)}$ for every $n>N_3$  for some large enough $N_3\in \NN$. Clearly this holds eventually, since $r(n)$ decreases exponentially. However, we also need a bound on the size of $N_3$, which requires more effort. Recall that $r(n)=\sum_{l=1}^m D_l \beta^n_l+ \co{D_l}\co \beta^n_l$. By applying The root separation bound~(\ref{eq:Mignotte}) from Section~\ref{sec:algebraic numbers} to $1-|\beta_l|$, we compute $\epsilon\in (0,1)$ and $N_3\in \NN$ such that $\frac{1}{\epsilon}$ and $N_3$ are $2^{\norm{I}^{O(1)}}$, and for every $n>N_3$ it holds that $|r(n)|<(1-\epsilon)^n$. Using this, we can find $N_4\in \NN$ such that $N_4=2^{\norm{I}^{O(1)}}$ and $|r(n)|<\frac{1}{p(n)}$ for all $n>N_4$, from which we can conclude the proof.
	\end{proof-synopsis}

	We are now ready to prove Theorem~\ref{thm:solve system}
	\begin{proof}
	For every expression in ${\rm Sys}$, let $f$ be the corresponding function as per Lemma~\ref{lem:main lemma}, and compute its respective bound $N$. If $\relBow$ is ``$=$'', then by Lemma~\ref{lem:main lemma}, if the equation is satisfiable for $n\in \NN$, then $n<N$.

	If all the $\relBow$ are ``$>$'', then for each such inequality compute $\set{z: f(z)>0}$. If the intersection of these sets is empty, then if $n$ is a solution for the system, it must hold that $n<N$. If the intersection is non-empty, then it is an open set. Since $\gamma$ is not a root of unity, then $\set{\gamma^n:n\in \NN}$ is dense in the unit circle. Thus, there exists $n>N$ such that $\gamma^n$ is in the above intersection, so the system has a solution. Checking the emptiness of the intersection can be done in polynomial time using Theorem~\ref{thm:renegar}.
	
	Thus, it remains to check whether there exists a solution $n<N$. Recall that $N=2^{\norm{I}^{O(1)}}$. Thus, in order to check whether the system is solved for $n<N$, we need to compute, e.g., $\alpha^{2n}$, whose representation is exponential in $\norm{I}$, so a naive implementation would take exponential space.
		
	Instead, we take a similar approach to~\cite{chonev2015polyhedron}: by representing numbers as arithmetic circuits, deciding the positivity (or testing for $0$ equality) can be done using an oracle to \posslp, which by~\cite{allender2009complexity} is in the counting hierarchy. By first guessing $n<N$, the problem can be solved in \NPposslp, which is contained in \PSPACE.
	\end{proof}
	

%% file: Conclusions.tex
\section{Conclusions}
\subsection{Proof of Theorem~\ref{thm:main PSPACE}}
We conclude by giving an explicit proof of Theorem~\ref{thm:main PSPACE}: Given polytopes $P_1$ and $P_2$ and a matrix $\A$, if $\A$ is singular, we first apply (in polynomial time) the reduction in Section~\ref{sec:reduction to invertible}. Thus, we can assume $\A$ is invertible. Next, if $P_1$ or $P_2$ are unbounded, for each unbounded face $F$ we proceed as follows: decompose $F$ as per Lemma~\ref{lem:2d polytope decomposition}, so $F=\bigcup_{i=1}^m A_i$, and recall that iterating over the $A_i$'s can be done in $\PSPACE$. 
In each iteration, consider an edge $E$ of $P_1$ and a face $F$ of $P_2$ (both of which may belong to sets $A_i$ as above). Formulate the first-order sentence~(\ref{eq:quantified}) in Section~\ref{sec:from invertible to equations}, and apply Theorem~\ref{thm:systems} to obtain an equivalent disjunction of systems $\bigvee_{i=1}^M{\rm Sys_i}$, where $M$ is constant. Then, for each system ${\rm Sys_i}$, check in $\PSPACE$ whether it has a solution, using either Section~\ref{sec:root of unity} or Theorem~\ref{thm:solve system}. 
If no solution was found, check in $\PSPACE$ whether a vertex of $P_1$ collides with $P_2$, using the algorithm in~\cite{chonev2015polyhedron}. Then, if still no solution is found, repeat the same procedure by interchanging the roles of $P_1$ and $P_2$, and considering the matrix $\A^{-1}$ instead of $\A$. The correctness and complexity of this procedure follow from the proofs of the respective theorems. 

\subsection{Discussion}
This paper studies an extension of the Orbit Problem, in which the input is existentially quantified over a polytope, and the target is a polytope. 
The importance of this work is twofold: from a practical perspective, we provide an algorithm for deciding the termination of linear while loops with affine guards, up to dimension three, when the input is not fixed. From a more theoretical perspective, and as already pointed out by Kannan and Lipton in~\cite{kannan1986polynomial}, the Orbit Problem and its variants are closely related to long-standing open problems such as the {\em Skolem Problem}, and various number-theoretic problems. It is therefore useful and compelling to push the borders of decidability, in order to identify the core of the remaining difficulties, and to eventually hopefully overcome them. 

Finally, as discussed in Section~\ref{sec: problem intro}, the problem at hand can be viewed as a particular case of the Orbit Problem in dimension six where the target is a semi-algebraic set. As the general problem is known to be hard even in dimension four, our work here suggests that interesting and useful fragments are tractable even in high dimensions.

%% file: appendix.tex
\section{Reduction to the Invertible Case}
\label{apx:reduction to invertible}
Recall that we are given polytopes $P,R\subseteq\RR^3$ and a matrix $\A\in (\AA\cap \RR)^{3\times 3}$, where $0$ is an eigenvalue of $\A$ with multiplicity $1$. As discussed in Section~\ref{sec:reduction to invertible}, we can write $\A=D^{-1} \begin{pmatrix}
0 & 0 \\
0 & B \\
\end{pmatrix} D$ where $D$ is an invertible matrix with real-algebraic entries, and $B\in (\AA\cap \RR)^{2\times 2}$ is also invertible.

Then, for every $x\in P$ and $n\in \NN$ it holds that $\A^nx\in R$ iff $D^{-1}\begin{pmatrix}
	0 & 0 \\
	0 & B^n \\
\end{pmatrix}Dx\in R$ iff $\begin{pmatrix}
0 & 0 \\
0 & B^n \\
\end{pmatrix}Dx\in DR$ (where $DR=\set{Dv:v\in R}$).
Observe that the first coordinate of $\begin{pmatrix}
0 & 0 \\
0 & B \\
\end{pmatrix}Dx$ is  $0$ for every vector $x$, and Consider the set $\begin{pmatrix}
0 & 0 \\
0 & B \\
\end{pmatrix}DP=\set{\begin{pmatrix}
	0 & 0 \\
	0 & B \\
	\end{pmatrix}Dx: x\in P}$. We can write 
	$$\begin{pmatrix}
0 & 0 \\
0 & B \\
\end{pmatrix}DP=\set{\begin{pmatrix}
	0\\
	x_1\\
	x_2
	\end{pmatrix}: \begin{pmatrix}
	x_1\\
	x_2
	\end{pmatrix}\in P'}$$ where $P'\subseteq \RR^2$ is the intersection of $\begin{pmatrix}
0 & 0 \\
0 & B \\
\end{pmatrix}DP$
with the $[yz]$ plane, in the standard basis $\set{\begin{pmatrix}
	1\\0
	\end{pmatrix},\begin{pmatrix}
	0\\1
	\end{pmatrix}}$.

Now, for $n\ge 1$, we get that there exists $x\in P$ such that $\A^n x\in R$ iff there exists $x\in P$ such that 
$$\begin{pmatrix}
0 & 0 \\
0 & B^{n-1} \\
\end{pmatrix}\begin{pmatrix}
0 & 0 \\
0 & B \\
\end{pmatrix}Dx \in DR$$ iff there exists $x'\in P'$ such that $B^{n-1}x'\in \begin{pmatrix}
0 & 1 & 0\\
0 & 0 & 1
\end{pmatrix}\left(DR\cap sp(\set{\begin{pmatrix}
	0\\1\\0
	\end{pmatrix},\begin{pmatrix}
	0\\0\\1
	\end{pmatrix}})\right)$.

Since all the intersections and matrices above can be computed in polynomial time, and since the intersections above are polytopes, we conclude that if $\A$ is singular, we can reduce the dimension of the problem. 

Next, if the multiplicity of $0$ is $2$, then we can write  $\A=D^{-1}\begin{pmatrix}
0 & 1 & 0\\
0 & 0 & 0\\
0 & 0 & \rho
\end{pmatrix} D$ where $\rho$ is a real eigenvalue. Then $\A^n=D^{-1}\begin{pmatrix}
0 & 0 & 0\\
0 & 0 & 0\\
0 & 0 & \rho^n
\end{pmatrix} D$ for every $n\ge 2$, and the same approach as above can be taken.

Finally, if the multiplicity of $0$ is $3$, then $\A^3=0$, so the problem becomes trivial.

\section{Quantifier Elimination}
\label{apx:quantifier elimination}
In this section we eliminate quantifiers from the expression
\begin{equation}
\label{eq:quantifiedApx}
\exists \lambda,\mu,\nu: 0\le \lambda,\mu,\nu\le 1 \wedge \mu+\nu\le 1 \wedge \bigwedge_{i=1}^3 \left( f_i\alpha^n +\co{f_i}\, \co{\alpha}^n + g_i \rho^n + \lambda (h_i\alpha^n +\co{h_i}\, \co{\alpha}^n + k_i\rho^n)= s_i +\mu t_i +\nu r_i\right)
\end{equation}
using the Fourier-Motzkin quantifier-elimination algorithm. 

We start by recalling the Fourier-Motzkin algorithm. Given a set of linear inequalities in the variables $x$ (which we want to eliminate), isolate $x$ in each equation. Then, for each pair of equations of the form $x\le {\rm expression}_1$ and $x\ge {\rm expression}_2$, add the inequality ${\rm expression}_1\ge {\rm expression}_2$ (with analogous rules for strict inequalities). After doing so for every relevant pair of inequalities, remove all original inequalities involving $x$. The Fourier-Motzkin Theorem states that the new system is satisfiable iff the original system is satisfiable. Note that the new system is also a system of linear inequalities in the original variables.

By repeating this process for all variables, we end up with an equivalent, variable-free system of inequalities. 

For the purpose of proving Theorem~\ref{thm:systems}, we need some assumptions on the coefficients of the resulting inequalities, in order to have the form described in Theorem~\ref{thm:systems}. We thus analyze in some detail the specific application of Fourier-Motzkin elimination to our setting.

\subsection{Proof of Theorem~\ref{thm:systems}}
\label{sec:proof of Theorem Systems}
We start by explicitly writing down the expressions we consider in~\ref{eq:quantifiedApx}. We think of ``$=$'' as a pair of ``$\ge$'' and ``$\le$'' inequalities.

$$
\begin{array}{l}
\lambda \leq 1 \\
\lambda \geq 0 \\
\mu \leq 1 \\
\mu \geq 0 \\
\nu \leq 1 \\
\nu \geq 0 \\
\mu +\nu \leq 1 \\
\co{\alpha} ^n \co{f_1}+\alpha ^n f_1+\rho ^n g_1+\lambda  \left(\co{\alpha} ^n \co{h_1}+\alpha ^n h_1+\rho ^n k_1\right)-(\nu  r_1+s_1+\mu  t_1)=0 \\
\co{\alpha} ^n \co{f_2}+\alpha ^n f_2+\rho ^n g_2+\lambda  \left(\co{\alpha} ^n \co{h_2}+\alpha ^n h_2+\rho ^n k_2\right)-(\nu  r_2+s_2+\mu  t_2) =0\\
\co{\alpha} ^n \co{f_3}+\alpha ^n f_3+\rho ^n g_3+\lambda  \left(\co{\alpha} ^n \co{h_3}+\alpha ^n h_3+\rho ^n k_3\right)-(\nu  r_3+s_3+\mu  t_3)=0 \\
\end{array}
$$

We make the following observations on the structure of the system.
\begin{observation}
	\label{obs:coefficients}
	The coefficients of the system above satisfy the following.
	\begin{enumerate}
	\item \label{item:nu} The coefficients of $\nu$ do not depend on $\alpha,\rho$ or $n$.
	\item \label{item:mu}The coefficients of $\mu$ do not depend on $\alpha,\rho$ or $n$.
	\item \label{item:lambda}The coefficients of $\lambda$ are either constant, or of the form $A\alpha^n+\co{A}\co{\alpha}^n+B\rho^n$, for some $A\in \AA$ and $B\in \RR\cap \AA$ (that is, the coefficients of $\alpha^n$ and $\co{\alpha}^n$ are conjugates, and the coefficient of $\rho^n$ is real)
	\item \label{item:free}The free coefficients of the form $A\alpha^n+\co{A}\co{\alpha}^n+B\rho^n+C$, for some $A\in \AA$ and $B,C\in \RR\cap \AA$
	\end{enumerate}
\end{observation}
We eliminate $\nu$ first. By Observation~\ref{obs:coefficients}.\ref{item:nu}, after isolating $\nu$ (which involves dividing by the coefficient of $\nu$), Observations \ref{obs:coefficients}.\ref{item:mu}, \ref{obs:coefficients}.\ref{item:lambda}, and \ref{obs:coefficients}.\ref{item:free} still hold. Thus, after eliminating $\nu$ and aggregating the coefficients of $\mu$ and $\lambda$, Observations \ref{obs:coefficients}.\ref{item:mu}, \ref{obs:coefficients}.\ref{item:lambda}, and \ref{obs:coefficients}.\ref{item:free} still hold, and Observation Observation~\ref{obs:coefficients}.\ref{item:nu} is irrelevant, since $\nu$ was eliminated.

By Observation Observation~\ref{obs:coefficients}.\ref{item:mu}, following the same reasoning for eliminating $\mu$ results in a system of inequalities in $\lambda$ that satisfies Observations \ref{obs:coefficients}.\ref{item:lambda} and \ref{obs:coefficients}.\ref{item:free}.

It now remains to eliminate $\lambda$. Note that here, even isolating $\lambda$ is not trivial. Indeed, in order to divide by a coefficient $A\alpha^n+\co{A}\co{\alpha}^n+B\rho^n$, we need to know its sign (and whether it is $0$). Thus, at this point in the elimination, we split the system into a disjunction of systems, where in each system we add an assumption on the sign of $A\alpha^n+\co{A}\co{\alpha}^n+B\rho^n$. Thus, an inequality of the form $(A\alpha^n+\co{A}\co{\alpha}^n+B\rho^n)\lambda\le {\it expression}$ will yield a disjunction of three systems:
\begin{itemize}
	\item $\lambda \le \frac{{\it expression}}{A\alpha^n+\co{A}\co{\alpha}^n+B\rho^n}\wedge A\alpha^n+\co{A}\co{\alpha}^n+B\rho^n>0$ 
	\item $\lambda \ge \frac{{\it expression}}{A\alpha^n+\co{A}\co{\alpha}^n+B\rho^n}\wedge A\alpha^n+\co{A}\co{\alpha}^n+B\rho^n<0$
	\item $0 \le {\it expression}\wedge A\alpha^n+\co{A}\co{\alpha}^n+B\rho^n=0$
\end{itemize}

After constructing these systems and combining the inequalities according to the algorithm, we multiply by a common denominator to get a system of inequalities without variables. In these inequalities, we multiply expressions of the form of Observation \ref{obs:coefficients}.\ref{item:free} by either constants, or by expressions of the form of Observation 3. Thus, end up with expressions of either the form of Observation \ref{obs:coefficients}.\ref{item:free}, or of the form
\begin{align*}
&(A\alpha^n+\co{A}\co{\alpha}^n+B\rho^n)(A'\alpha^n+\co{A'}\co{\alpha}^n+B'\rho^n+C')=
AA'\alpha^{2n}+\co{AA'}\co{\alpha}^{2n}+\\
&(AB'+A'B) \alpha^n\rho^n+ (\co{A}B'+\co{A'}B) \co{\alpha^n}\rho^n+
BB'\rho^{2n}+(A\co{A'}+\co{A}A')|\alpha|^{2n}+C'A\alpha^n+C'\co{A}\co{\alpha}^n+C'B\rho^n
\end{align*}

Finally, we renaming the coefficients, and by adding a constant term, both the latter form and that of Observation \ref{obs:coefficients}.\ref{item:lambda} are as described in Theorem~\ref{thm:systems}. Finally, we split every nonstrict inequality to a disjunction of an equality and a strict inequality, and distribute the conjunction over them.

We note that the numbers of systems and equations are bounded by constants, since the removal does not depend on the coefficients, but only on the form of the expressions.
\qed

\section{The case of only real eigenvalues}
\label{apx:real eigen}
In this section we consider the case where the matrix $\A$ has only real eigenvalues, denoted $\rho_1,\rho_2,\rho_3$. In this case, by converting $\A$ to Jordan normal form, there exists an invertible matrix $B\in (\AA\cap \RR)^{3\times 3}$ such that one of the following holds:
\begin{enumerate}
\item $\A=B^{-1}\begin{pmatrix}
\rho_1 & 0 & 0\\
0 & \rho_2 & 0\\
0& 0 & \rho_3 \\
\end{pmatrix}B$, in which case $\A^n=B^{-1}\begin{pmatrix}
\rho_1^n & 0 & 0\\
0 & \rho_2^n & 0\\
0& 0 & \rho_3^n \\
\end{pmatrix}B$.
\item $\A=B^{-1}\begin{pmatrix}
\rho_1 & 1 & 0\\
0 & \rho_2 & 0\\
0& 0 & \rho_3 \\
\end{pmatrix}B$ with $\rho_1=\rho_2$, in which case $\A^n=B^{-1}\begin{pmatrix}
\rho_1^n & n\rho_1^{n-1} & 0\\
0 & \rho_1^n & 0\\
0& 0 & \rho_3^n \\
\end{pmatrix}B$.
\item $\A=B^{-1}\begin{pmatrix}
\rho_1 & 1 & 0\\
0 & \rho_2 & 1\\
0& 0 & \rho_3 \\
\end{pmatrix}B$ with $\rho_1=\rho_2=\rho_3$, in which case $\A^n=B^{-1}$\\ $\begin{pmatrix}
\rho_1^n & n\rho_1^{n-1} & \frac12n(n-1)\rho_1^{n-2}\\
0 & \rho_1^n & n\rho_1^{n-1}\\
0& 0 & \rho_1^n \\
\end{pmatrix}B$.
\end{enumerate}

We consider here the latter case, as the first two are similar and simpler.
We start by following the lines of Section~\ref{sec:from invertible to equations}. That is, we formulate the problem as a first-order sentence, and proceed to remove the quantifiers as per Appendix~\ref{apx:quantifier elimination}. Consider $n\in \NN$ and a vector $v$, then we can write
$$\A^nv=B^{-1}\begin{pmatrix}
\rho_1^n & n\rho_1^{n-1} & \frac12n(n-1)\rho_1^{n-2}\\
0 & \rho_1^n & n\rho_1^{n-1}\\
0& 0 & \rho_1^n \\
\end{pmatrix}Bv= a\rho_1^n+bn\rho_1^{n-1}+cn(n-1)\rho_1^{n-2}$$
Thus, the formulation of the first-order sentence~\ref{eq:quantified} in Section~\ref{sec:from invertible to equations} takes a similar form in this case, and after applying quantifier elimination, we end up with a disjunction as per Theorem~\ref{thm:systems}, where the expressions in each system are of the form
$$A \rho_1 ^{2n}+ Bn\rho_1^{2n}+Cn^2\rho_1^{2n}+ D\rho_1^n+ E n\rho_1^n+ Fn^2\rho_1^n+G\relBow 0$$
Assuming $\rho_1>0$ (otherwise we can split according to odd and even $n$), for each such expression we can compute a bound $N\in \NN$ based on the rate of growth of the different components, such that either for every $n>N$ the equation holds, or for every $n>N$ it does not hold. This is done in a similar manner to the proof of Theorem~\ref{thm:solve system}. Thus, either we determine that a solution exists since all the expressions are satisfied for large enough $n$, or we need to check the solutions up to $N$, which can be done in $\PSPACE$ (as in the proof of Theorem~\ref{thm:solve system}).

\section{The case where $\frac{\alpha}{|\alpha|}$ is a root of unity}
\label{apx:root of unity}
Let $\gamma=\frac{\alpha}{|\alpha|}$. We assume that $\gamma$ is a root of unity. Thus, there exists $d\in \NN$ such that $\gamma^d=1$. After obtaining the systems of expressions as per Theorem~\ref{thm:systems}, each expression~\ref{eq: generic expression} can be written as
$$|\alpha|^{2n}A\gamma^{2n}+|\alpha|^{2n}\co{A}\co{\gamma}^{2n}+ |\alpha|^{n}B\gamma^n\rho^n + |\alpha|^{n}\co{B}\co{\gamma}^n\rho^n + C\rho^{2n} +  D|\alpha|^{2n}+ |\alpha|^n E\gamma^n + |\alpha|^n\co{E}\co{\gamma}^n+ F\rho^n + G\relBow 0$$
Observe that $\gamma^2$ is also a root of unity of order at most $d$. Thus, for every $n\in \NN$ it holds that $(\gamma^n,\gamma^{2n})=(\gamma^{n\bmod d},\gamma^{2n\bmod d})$.  Consider the set $V=\set{(n\bmod d,2n\bmod d): n\in \NN}$, and note that $|V|\le d^2$. For every $(k,k')\in V$, let $N_{(k,k')}$ be the minimal number such that $(n\bmod d,2n\bmod d)=(k,k')$. Observe that $\{n\in \NN: (n\bmod d,2n\bmod d)=(k,k')\} = \{N_{(k,k')}+m|V|:m\in \NN\}$.
For each system ${\rm Sys}$ of expressions, we construct $|V|$ systems $\set{{\rm Sys}_{(k,k')}}_{(k,k')\in V}$ such that ${\rm Sys}_{(k,k')}$ is obtained from ${\rm Sys}$ by replacing, in every expression, $\gamma^n$ with $\gamma^{k}$, replacing $\gamma^{2n}$ with $\gamma^{k'}$, and replacing $n$ in the remaining powers by $N_{(k,k')}+m|V|$. By pushing constants into the coefficients and renaming $\alpha^|V|=\beta$ and $\rho^|V|=\delta$, the expression above can be written as
$$|\beta|^{2m}A\gamma^{k'}+|\beta|^{2m}\co{A}\co{\gamma}^{k'}+ |\beta|^{m}B\gamma^k\delta^{m} + |\beta|^{m}\co{B}\co{\gamma}^k\delta^{m} + C\delta^{2m} +  D|\beta|^{2m}+ |\beta|^{m} E\gamma^k + |\beta|^{m}\co{E}\co{\gamma}^k+ F\delta^{m} + G\relBow  0$$
This becomes
$$2\re(A\gamma^{k'})|\beta|^{2m}+ 2\re(B\gamma^k)|\beta|^{m}\delta^{m} + C\delta^{2m} + D|\beta|^{2m}+ 2\re(E\gamma^k)|\beta|^{m} + F\delta^{m} + G\relBow  0$$
These expressions contain only real-algebraic constants, and thus the system can be solved in similar techniques as those of Appendix~\ref{apx:real eigen}. 

Finally, we show that the number of systems is polynomial, by showing that $d\le \deg(\gamma)^2$. The proof appears in~\cite{kannan1986polynomial}, and we bring it here for completeness. Since $\gamma$ is a primitive root of unity of order $d$, then the defining polynomial $p_\gamma$ of $\gamma$ is the $d$-th Cyclotomic polynomial, so $\deg(\gamma)=\Phi(d)$, where $\Phi$ is Euler's totient function. Since $\Phi(d)\ge \sqrt{d}$, we get that $d\le \deg(\gamma)^2$.

\section{Proof of Lemma~\ref{lem:main lemma}}
\label{apx:proof main lemma}
By identifying $\CC$ with $\RR^2$ (where $z=x+iy$ is identified with $(x,y)$) we identify $f$ with the function $f:\RR^2\to \RR$ defined by
\begin{align*}
f(x,y)=&A(x+iy)^2+\co{A(x+iy)^{2}}+ B (x+iy) + \co{B(x+iy)} + C\\
=& 2\re(A(x+iy)^2)+2\re(B (x+iy))+C\\
=& 2\re((\re(A)+i\im(A))(x^2-y^2+i2xy))+2\re((\re(B)+i\im(B)) (x+iy))+C\\
=& 2(\re(A)(x^2-y^2)-\im(A)2xy)+2(\re(B)x-\im(B)y)+C
\end{align*}		
Since $\set{\gamma^n:n\in \NN}$ is dense on the unit circle, our interest in $f$ is also about the unit circle. 		
Since $f$ is a polynomial with algebraic coefficients, we can find in polynomial time a description of the algebraic set $\set{(x,y):f(x,y)=0 \wedge x^2+y^2=1}$. Note that since the coefficients of $x^2$ and $y^2$ in $f$ are either both $0$, or they differ in their sign, then this set is not the entire unit circle. Therefore, by B\'ezout's Theorem, this set is discrete and consists of at most $4$ points. Indeed, this set is the intersection of distinct contours of bivariate quadratic polynomials, so it corresponds to the roots of a polynomial of degree at most $4$. 
Let $\set{(x_1,y_1),\ldots,(x_4,y_4)}$ be these points, and let $z_1=x_1+iy_1, \ldots,z_4=x_4+iy_4$ be the respective complex numbers. Note that these points have algebraic coordinates, so $z_1,...,z_4$ are algebraic numbers.
Moreover, since these numbers are attained as the roots of a polynomial of degree 4 whose coefficients are polynomial in those of $f$, then we have that $\norm{z_1},\ldots,\norm{z_4}$ are polynomial in $\norm{f}$.
Note that if $A=B=0$, then $C\neq 0$ by our assumption, and there are no roots. Thus, we assume for now that $A$ and $B$ are not both $0$. We remove this assumption after we are done handling the roots.

Since $\gamma$ is not a root of unity, then in particular, for every $n_1\neq n_2\in \NN$ it holds that $\gamma^{n_1}\neq \gamma^{n_2}$. 
Thus, there exists $N_1\in\NN$ such that $\gamma^n\notin\set{z_1,\ldots,z_4}$ for every $n>N_1$. Moreover, by~\cite{chonev2013complexity}, we have that $N_1=k^{O(1)}$, where  $k= \norm{\gamma}+\sum_{j=1}^4 \norm{z_j}$, and $N_1$ can be computed in polynomial time in $k$.
Then, by Lemma~\ref{lem:Baker on unit circle}, there exists a constant $D\in \NN$ such that for every $n\ge N_1$ and $1\le j\le 4$ we have that $|\gamma^n-z_j|>\frac{1}{n^{(k^D)}}$. 

Let $\theta_A=\arg (A)$, $\theta_B=\arg (B)$, and $\phi_j=\arg(z_j)$ for every $1\le j\le 4$. We assume w.l.o.g.\ that the angles $\phi_j$ are all in $(-\pi,\pi)$. Otherwise (if one of the angles is exactly $\pi$), we shift the domain such that all the angles are in the interior. 
Define $g:(-\pi,\pi] \to \RR$ by $g(x)=f(e^{ix})$, so that $g(x)=2|A|\cos(2x+\theta_A)+2|B|\cos(x+\theta_B)+C$. Our next step is to show that $|g|$ is bounded from below by a polynomial. More precisely, we will show that $|g|$ is bounded from below in neighbourhoods of the roots of $g$, and give a lower bound on the value of $|g|$ outside these neighbourhoods. Technically, we will use the Taylor polynomials of $g$ to obtain these bounds.
For every $1\le j\le 4$, let $T_j$ be the Taylor polynomial of $g$ around $\phi_j$ such that the degree $d_j$ of $T_j$ is minimal and $T_j$ is not identically 0. Thus, we have $T_j(x)=\frac{g^{(d_j)}}{d_j!}(x-\phi)^{d_j}$. We now show that in fact, the degrees of these polynomials are at most three.
\begin{lemma}
\label{lem:degree of taylor}
$d_j\le 4$ for every $1\le j\le 3$.
\end{lemma}		
\begin{proof}
It is enough to show that at every point where $g(x)=0$, at least one of the first three derivatives of $g$ is non-zero. Assume by way of contradiction that the first three derivatives are all $0$ at $x$, and $g(x)=0$, then we have
\begin{align*}
g(x)=&2|A|\cos(2x+\theta_A)+2|B|\cos(x+\theta_B)=0\\
g'(x)=&-4|A|\sin(2x+\theta_A)-2|B|\sin(x+\theta_B)=0\\
g''(x)=&-8|A|\cos(2x+\theta_A)-2|B|\cos(x+\theta_B)=0\\
g^{(3)}(x)=&16|A|\sin(2x+\theta_A)+2|B|\sin(x+\theta_B)=0
\end{align*}
Pairing the odd and even derivatives, this can be written as
$$\begin{pmatrix}
2|A| & 2|B|\\
-8|A| & -2|B|
\end{pmatrix}\vect{\cos(2x+\theta_A)}{\cos(x+\theta_B)}=\vect{0}{0}
\text{ and }
\begin{pmatrix}
-4|A| & -2|B|\\
16|A| & 2|B|
\end{pmatrix}\vect{\sin(2x+\theta_A)}{\sin(x+\theta_B)}=\vect{0}{0}$$
If either $|A|$ or $|B|$ are 0 (but not both, as per our assumption above), then clearly either the first or second derivatives are always nonzero (since this is a single trigonometric function).
If $|A|,|B|\neq 0$, then the matrices are invertible, so it must hold that $\sin(2x+\theta_A)=\sin(x+\theta_B)=\cos(2x+\theta_A)=\cos(x+\theta_B)=0$, which clearly has no solution.
\end{proof}
Thus, $T_j(x)$ is a polynomial of degree at most three, with $T_j(\phi_j)=0$. We remark that $T_j$ is computable in polynomial time (in $\norm{f}$), as the coefficients are polynomial in $\norm{A},\norm{B},\norm{\gamma}$.

By Taylor's inequality, we have that for every $x\in [-\pi,\pi]$ it holds that $|g(x)-T_j(x)|\le \frac{M_j|x-\phi_j|^{d_j+1}}{(d_j+1)!}$ where $M_j=\max_{x\in [-\pi,\pi]}\set{g^{(d_j+1)}(x)}\le 32|A|+2|B|$ (where $g$ is extended naturally to the domain $[-\pi,\pi]$). 

Let $\epsilon_1>0$ such that the following hold for every $1\le j\le 4$.
\begin{enumerate}
\item $|g(x)-T_j(x)|\le \frac{1}{2}|T_j(x)|$ for every $x\in (\phi_j-\epsilon_1,\phi_j+\epsilon_1)$. 
\item ${\rm sign}(g'(x))$ does not change in $(\phi_j,\phi_j+\epsilon_1)$ nor in $(\phi_j-\epsilon_1,\phi_j)$.
\item ${\rm sign}(g'(x))={\rm sign}(T_j'(x))$ for every $x\in (\phi_j-\epsilon_1,\phi_j+\epsilon_1)$.
\end{enumerate}
Note that we can assume $(\phi_j-\epsilon_1,\phi_j+\epsilon_1)\subseteq (-\pi,\pi)$, since by our assumption $\theta_j\in (-\pi,\pi)$ for all $1\le j\le 4$.

An $\epsilon_1$ as above exists since $T_j(x)$ is of degree $d_j$, whereas the $|g(x)-T(x)|$ is of degree $d_j+1$, since there are only finitely many points where $g'(x)=0$, and since $T'(x)$ is the Taylor polynomial of degree $d_j-1$ of $g'(x)$ around $\phi_j$, so by bounding the distance $|g'(x)-T'(x)|$ we can conclude the third requirement (see Figure~\ref{fig:functions} for an illustration).
	\begin{figure}[ht]
		\centering
		\includegraphics[width=0.6\linewidth]{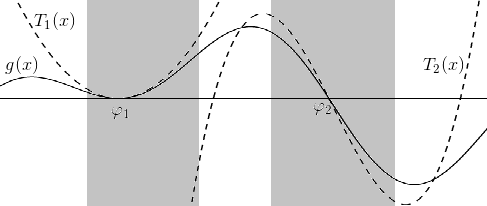}
		\caption{$g(x)$ and two Taylor polynomials: $T_1(x)$ around $\phi_1$ and $T_2(x)$ around $\phi_2$. The shaded regions show where requirements (1)--(3) hold, which determine $\epsilon_1$. Observe that for $T_1$, the most restrictive requirement is $|g(x)-T_1(x)|\le \frac12 T_1(x)$, whereas for $T_2$ the restriction is the requirement that $T_2(x)$ is monotone.}
		\label{fig:functions}
	\end{figure}
For the following, we need also to compute $\epsilon_1$, we thus proceed with the following lemma.
\begin{lemma}
\label{lem:epsilon computable}
$\epsilon_1$ can be computed in polynomial time in $\norm{f}$, and $\frac1\epsilon=2^{n^{O(1)}}$.
\end{lemma}
\begin{proof}
We start with Condition 2, and compute $\delta_1>0$ such that ${\rm sign}(g'(x))$ does not change in $(\phi_j-\delta_1,\phi_j)$ nor in $(\phi_j,\phi_j+\delta_1)$. This is done as follows. Recall that $g(x)=f(e^{ix})$, then we have $g'(x)=f'(e^{ix})ie^{ix}$. Since $f'$ is a polynomial with algebraic coefficients, then $F=\set{z: |z|=1\wedge f'(z)iz=0}$ consists of algebraic numbers whose degree and height are polynomial in those of $A$ and $B$, and we have that $\set{x: g'(x)=0}=\set{\arg(z):z\in F}$. By similar arguments as those by which we found the roots of $f$ on the unit circle, we can conclude that $F$ contains at most four points. Thus, it is enough to set $\delta_1$ such that $\left(\bigcup_{j=1}^4 (\phi_j-\delta_1,\phi_j)\cup (\phi_j,\phi_j+\delta_1)\right)\cap F=\emptyset$.
By~Equation (\ref{eq:Mignotte}), we have that for $z\neq z'\in F$ it holds that $|z-z'|>\frac{\sqrt{6}}{d^\frac{d+1}{2}\cdot H^{d-1}}$ where $d$ and $H$ are the degree and height of $f'(z)iz$. Thus, $1/|z-z'|$ is $2^{\norm{f}^{O(1)}}$, and has a polynomial description. Since $|\arg(z)-\arg(z')|>|z-z'|$, we conclude that by setting $\delta_1=\min{|z-z'|: z\neq z'\in F}/3$, and it holds that $\frac{1}{\delta_1}$ has a polynomial description in $\norm{f}$, and $\delta_1$ satisfies the required condition.

We now proceed to handle Condition 1, and compute $\delta_2>0$ such that $|g(x)-T_j(x)|\le \frac{1}{2}|T_j(x)|$ for every $x\in (\phi_j-\delta_2,\phi_j+\delta_2)$. Recall that $T_j(x)=\frac{g^{(d_j)}}{d_j!}(x-\phi_j)^{d_j}$.
Note that this case is more challenging than Condition 2, as unlike $g(x)=f(e^{ix})$, the polynomial $T_j(x)$ has potentially transcendental coefficients (namely $\phi_j$). In order to ignore the absolute value, assume  $T_j(x)\ge g(x)\ge \frac12 T_j(x)>0$ in an interval $(\phi_j,\phi_j+\xi)$ for some $\xi>0$ (the other cases are treated similarly). Then, the inequality above becomes $g(x)-\frac12 T_j(x)\ge 0$. 			
Since the degree of $T_j$ is $d_j$, then by the definition of $T_j$, the first $d_j-1$ derivatives of $g$ in $\phi_j$ vanish. Define $h(x)=g(x)-\frac12 T_j(x)$, then we have $h(\phi_j)=0$, $h'(\phi_j)=0,\ldots,h^{(d_j-1)}(\phi_j)=0$ and $h^{(d_j)}(\phi_j)=g^{(d_j)}(\phi_j)-\frac12 g^{(d_j)}(\phi_j)=\frac12 g^{(d_j)}(\phi_j)$. By our assumption, $T_j(x)\ge \frac12 T_j(x)$ for $x\in (\phi,\phi+\xi)$, so $h^{(d_j)}>0$.
In addition, recall that $|h^{(d_j+1)}(x)|=|g^{(d_j+1)}(x)|\le M_j\le 64|A|+2|B|$ for every $x\in [-\pi,\pi]$. Thus, by writing the $d_j$-th Taylor expansion of $h(x)$ around $\phi$, we have that $h(x)=h^{{d_j}}(\phi)(x-\phi)^{d_j}+R(x)$ where $|R(x)|\le \frac{M_j}{(d_j+1)!}(x-\phi)^{k+1}$. We thus have that for $x\in (\phi,\phi+\frac{g^{d_j}(\phi)(d_j+1)}{2 M_j})$ it holds that $h(x)\ge 0$. We can now set $\delta_2=\frac{g^{d_j}(\phi)(d_j+1)}{2 M_j}$, which satisfies the required condition (or a similar $\delta_2$ after analyzing the other cases).

Finally, we address Condition 3, and compute $\delta_3>0$ such that ${\rm sign}(g'(x))={\rm sign}(T'_j(x))$ for every $x\in (\phi_j-\delta_3,\phi_j+\delta_3)$. Observe that $T'_j(x)$ is the $d_j-1$-th Taylor polynomial of $g'(x)$ around $\phi$. Thus, by following the reasoning used to find $\delta_2$, we can find $\delta_3$ such that $|g'(x)-T'_j(x)|\le \frac12 |T_j(x)|$ for every $x\in (\phi-\delta_3,\phi+\delta_3)$, and in particular it holds that  ${\rm sign}(g'(x))={\rm sign}(T'_j(x))$ for every $x\in (\phi_j-\delta_3,\phi_j+\delta_3)$.

By setting $\epsilon_1=\min\set{\delta_1,\delta_2,
\delta_3}$, we conclude the proof.
\end{proof}

Conditions 1,2,3 above imply that within the intervals $(\phi_j-\epsilon_1,\phi_j+\epsilon_1)$ we have that $|g(x)|\ge \frac{1}{2}|T_j(x)|$, that $g(x)$ and $T_j(x)$ have the same sign, and that they are both decreasing/increasing together.

We now claim that there exists a polynomial $p(n)$ and a number $N_2\in \NN$ such that for every $n>N_2$ it holds that  $|g(\arg(\gamma^n))|>\frac{1}{p(n)}$. In order to compute $p(n)$, we compute separate polynomials for the domain $\bigcup_{j=1}^4 (\phi_j-\epsilon_1,\phi_j+\epsilon_1)$ and for its complement. Then, taking their minimum and bounding it with a polynomial yields $p(n)$.

At this point we also drop the assumption that either $A$ or $B$ are nonzero. Indeed, if $A=B=0$, then $C\neq 0$, and the above is trivial.

We start by considering the case where $\arg(\gamma^n)\in \bigcup_{j=1}^4 (\phi_j-\epsilon_1,\phi_j+\epsilon_1)$. Recall that since $\gamma$ is not a root of unity, then for every $n>N_1$ it holds that $\gamma^n\notin \set{z_1,\ldots,z_4}$. Then, by Lemma~\ref{lem:Baker on unit circle}, for every  $1\le j\le 4$ and every $n\ge N_2=\max\set{N_1,2}$ we have $|\gamma^n-z_j|> \frac{1}{n^{(k^D)}}$. In addition, $|\gamma^n-z_j|\le |\arg(\gamma^n)-\phi_j|$ (since the LHS is the Euclidean distance and the RHS is the spherical distance). Therefore, $|\arg(\gamma^n)-\phi_j|>\frac{1}{n^{(k^D)}}$, so either $\arg(\gamma^n)>\phi_j+\frac{1}{n^{(k^D)}}$ or $\arg(\gamma^n)<\phi_j-\frac{1}{n^{(k^D)}}$. Next, we have that if $\arg(\gamma^n)\in (\phi_j-\epsilon_1,\phi_j+\epsilon_1)$ for some $1\le j\le 4$, then $|g(\arg(\gamma^n))|\ge \frac12 |T_j(\arg(\gamma^n))|\ge \frac12\min\set{{|T_j(\phi_j+\frac{1}{n^{(k^D)}})|,|T_j(\phi_j-\frac{1}{n^{(k^D)}})|}}$, where the last inequality follows from condition 3 above, which implies that $T_j$ is monotone with the same tendency as $g$. 

Observe that $T_j(\phi_j-\frac{1}{n^{(k^D)}})=\frac{g^{(d_j)}(\phi)}{d_j!}\frac{1}{n^{(k^D)}}$ and similarly $T_j(\phi_j+\frac{1}{n^{(k^D)}})=-\frac{g^{(d_j)}(\phi)}{d_j!}\frac{1}{n^{(k^D)}}$ are both inverse polynomials.
Thus, $|g(\arg(\gamma^n))|$ is bounded from below by an inverse polynomial. Moreover, these polynomials can be easily computed in time polynomial in $\norm{f}$.

Finally, we note that for $x\notin \bigcup_{j=1}^4 (\phi_j-\epsilon_1,\phi_j+\epsilon_1)$ we can compute in polynomial time a bound $B>0$ such that $|g(x)|>B$. Indeed, $B=\min\set{|g(x)|: x\in [-\pi,\pi]\setminus\bigcup_{j=1}^4 (\phi_j-\epsilon_1,\phi_j+\epsilon_1)}$ (where $g(-\pi)$ is defined naturally by extending the domain), and we have that $|B|>0$ since we assumed non of the $\phi_j$ are exactly at $\pi$ (in which case we would have had $g(-\pi)=0$). In particular, we can combine the two domains and compute a polynomial $p$ as required.
We remark that we can compute $\norm{B}$ in polynomial time, since it is either at least $\frac12|T_j(\phi_j\pm \epsilon_1)|$ for some $1\le j\le 4$ (and by Lemma~\ref{lem:epsilon computable}, $\norm{\epsilon_1}$ can be computed in polynomial time), or it is the value of one of the extrema of $g$, and the latter can be computed by finding the extrema of the (algebraic) function $f$ on the unit circle.

To recap, for every $n>N_2$ it holds that $|g(\arg(\gamma^n))|>\frac{1}{p(n)}$ for a non-negative polynomial $p$, and both $N_2$ and $p$ can be computed in polynomial time in the description of the input.

Nest, we wish to find $N_3\in \NN$ such that for every $n>N_3$ it holds that $r(n)<\frac{1}{p(n)}$. Recall that 	$r(n)=\sum_{l=1}^m D_l \beta^n_l+ \co{D_l}\co \beta^n_l$. Let $1\le l\le m$, and consider $\beta_l$. 
Since $\beta_l$ is algebraic, then so is $1-|\beta_l|$. Indeed, $1-|\beta_l|=1-\sqrt{\beta_l\co{\beta_l}}$. Moreover, we get that $\deg(1-|\beta_l|)\le \deg(\beta_l)^4$ the root of a polynomial of degree at most $\deg(\beta_l)^4$, 
and of height polynomial in $\heig{\beta_l}$.  
Since $|\beta_l|<1$, By applying Equation~\ref{eq:Mignotte}, we get $1-|\beta_l|=|1-|\beta_l||>\frac{\sqrt{6}}{d^{(d+1)/2} \heig{\beta_l}^{d-1}}$ where $d=\deg(\beta_l)^{O(1)}$. Recall that $\heig{\beta_l}=2^{\norm{I}^{O(1)}}$. Thus, we can compute $\epsilon\in (0,1)$ and $N_3\in \NN$ such that:
\begin{enumerate}
\item $\frac{1}{\epsilon}=2^{\norm{I}^{O(1)}}$
\item $N_3=2^{\norm{I}^{O(1)}}$
\item For every $n>N_3$ it holds that $|r(n)|<(1-\epsilon)^n$
\end{enumerate}

Finally, by taking $N_4\in \NN$ such that $(1-\epsilon)^n<\frac{1}{p(n)}$ (which satisfies $N_4=2^{\norm{I}^{O(1)}}$) for all $n>N_4$, we can now conclude that for every $n>\max\set{N_2,N_3,N_4}$, the following hold.
\begin{enumerate}
\item $f(\gamma^n)=g(\arg(\gamma^n))\neq 0$.
\item If $f(\gamma^n)>0$, then $g(\arg(\gamma^n))>0$, so $g(\arg(\gamma^n))>\frac{1}{p(n)}$. Since $|r(n)|<\frac{1}{p(n)}$, it follows that $f(\gamma^n)+r(n)=g(\arg(\gamma^n))+r(n)>\frac{1}{p(n)}-|r(n)|>0$. Conversely, if $f(\gamma^n)+r(n)>0$, then $g(\arg(\gamma^n))+r(n)>0$, but sinc e $|g(\arg(\gamma^n))|>\frac{1}{p(n)}$ and $|r(n)|<\frac{1}{p(n)}$, then it must hold that $g(\arg(\gamma^n))>0$, so $f(\gamma^n)>0$.
\item If $f(\gamma^n)<0$, then $g(\arg(\gamma^n))<0$, so $g(\arg(\gamma^n))<-\frac{1}{p(n)}$. Since $|r(n)|<\frac{1}{p(n)}$, it follows that $f(\gamma^n)+r(n)=g(\arg(\gamma^n))+r(n)<-\frac{1}{p(n)}+|r(n)|<0$. Conversely, if $f(\gamma^n)+r(n)<0$, then $g(\arg(\gamma^n))+r(n)<0$, but since $|g(\arg(\gamma^n))|>\frac{1}{p(n)}$ and $|r(n)|<\frac{1}{p(n)}$, then it must hold that $g(\arg(\gamma^n))<0$, so $f(\gamma^n)<0$.
\end{enumerate}
Which concludes the proof of Lemma~\ref{lem:main lemma}.
\qed